\documentclass[12pt]{article}
\usepackage{graphicx}
\usepackage{natbib} 
\usepackage{url} 
\usepackage{bm}
\usepackage{xcolor}
\usepackage{comment}
\usepackage{caption}
\usepackage{subcaption}
\usepackage{xcolor}
\usepackage{hyperref}
\usepackage{caption}
\usepackage{subcaption}
\usepackage{ulem}
\usepackage{graphicx,psfrag,epsf}
\usepackage{enumerate}
\usepackage{natbib}
\usepackage{url} 
\usepackage{amsmath}
\usepackage{graphicx,psfrag,epsf}
\usepackage{enumerate}
\usepackage{natbib}
\usepackage{url} 
\usepackage{bm}
\usepackage{amsfonts}
\usepackage{comment}
\usepackage{xcolor}
\usepackage{comment}
\usepackage{caption}
\usepackage{subcaption}
\usepackage{xcolor}
\usepackage{caption}
\usepackage{subcaption}
\usepackage{ulem}
\newcommand{\code}[1]{\texttt{#1}}

\newcommand{\blind}{0}

\usepackage[english]{babel}
\usepackage{amsthm}
\newtheorem{theorem}{Theorem}
\newtheorem{lemma}[theorem]{Lemma}
\newtheorem{remark}{Remark}
\newtheorem{example}{Example}
\newtheorem{definition}{Definition}
\begin{document}

\def\spacingset#1{\renewcommand{\baselinestretch}%
{#1}\small\normalsize} \spacingset{1}


\if0\blind
{
  \title{\bf Spatial von-Mises Fisher Regression for Directional Data}
  \author{Zhou Lan\thanks{zlan@bwh.harvard.edu
    }\hspace{.2cm}\\
    Brigham and Women's Hospital, Harvard Medical School\\
    and \\
    Arkaprava Roy\thanks{arkaprava.roy@ufl.edu
    }\hspace{.2cm} \\
    Department of Biostatistics, University of Florida\\
    and\\
    For The Alzheimer’s Disease Neuroimaging Initiative\thanks{Data used in the preparation of this article were obtained from the Alzheimer's Disease Neuroimaging Initiative
(ADNI) database (adni.loni.usc.edu). As such, the investigators within the ADNI contributed to the design
and implementation of ADNI and/or provided data but did not participate in analysis or writing of this report.
A complete listing of ADNI investigators can be found at:
\url{http://adni.loni.usc.edu/wp-content/uploads/how_to_apply/ADNI_Acknowledgement_List.pdf}}
    }
  \maketitle
} \fi

\if1\blind
{
  \bigskip
  \bigskip
  \bigskip
  \begin{center}
    {\LARGE\bf Spatial von-Mises Fisher Regression for Directional Data}
\end{center}
  \medskip
} \fi

\begin{abstract}
Spatially varying directional data are routinely observed in several modern applications such as meteorology, biology, geophysics, engineering, etc. However, only a few approaches are available for covariate-dependent statistical analysis for such data. To address this gap, we propose a novel generalized linear model to analyze such that using a von Mises Fisher (vMF) distributed error structure. Using a novel link function that relies on the transformation between Cartesian and spherical coordinates, we regress the vMF-distributed directional data on the external covariates. This regression model enables us to quantify the impact of external factors on the observed directional data. Furthermore, we impose the spatial dependence using an autoregressive model, appropriately accounting for the directional dependence in the outcome. This novel specification renders computational efficiency and flexibility. In addition, a comprehensive Bayesian inferential toolbox is thoroughly developed and applied to our analysis. Subsequently, employing our regression model on the Alzheimer's Disease Neuroimaging Initiative (ADNI) data, we gain new insights into the relationship between cognitive impairment and the orientations of brain fibers, along with examining empirical efficacy through simulation experiments. The code for implementing our proposed method is available on GitHub: \url{https://github.com/lanzhouBWH/Spatial_VMF_Regression}.
\end{abstract}

\noindent%
\vfill

\newpage
\spacingset{1.8} 

\section{Introduction}
\label{sec:intro}
Directional statistics, alternatively referred to as circular or angular statistics in two-dimensional cases or spherical statistics for higher dimensions, is a distinct subfield of multivariate statistics tailored explicitly for the analysis of directional data \citep{mardia2009directional}. Such data arise when measurements are expressed as angles or directions, inherently exhibiting spherical relationships. Instances of directional data can be found in diverse disciplines, including meteorology, biology, geophysics, and engineering, with examples ranging from wind direction and animal migration patterns to diffusion tensor imaging. In addition, a variety of external features often exist that control the directional outcome.

In this paper, our proposed methodological developments focus on studying the association between white matter fiber orientations in the brain and subject-level covariates (e.g., age, sex, cognitive score, and genetic information), revealing the factors driving fiber orientations \citep{schwartzman2005cross,schwartzman2008inference}.
In brain neuroimaging analyses, spatial dependencies become an important feature. Furthermore, the incorporation of spatial dependence is appealing in many directional data applications such as wind direction modeling \citep{neto2014accounting,murphy2022joint}, brain fiber orientation modeling \citep{goodlett2009group, zhu2011fadtts}, and modeling actigraph data from wearable devices \citep{banerjee2023bayesian}. Hence, to effectively analyze this type of data, it is important to employ a strategy that can address a) the inherent directional nature of the outcome; b) the potential spatial dependencies; and c) an effective and convenient summarization tool for the effects of the external factors associated with the outcome.

The fiber orientation data considered in this paper are derived from diffusion tensor imaging (DTI), a technique that has gained much popularity for its increasing efficiency and accuracy in measuring microstructures \citep{soares2013hitchhiker}. In the DTI imaging technique, a diffusion ellipsoid is commonly used, which represents diffusivity in each voxel. Here, we focus on the principal eigenvector $\bm{E}_1$, interpreted as the tangent direction along the fiber bundle (rightmost panel in Figure \ref{fig:fiber}). Thus, it is referred to as the \textit{principal diffusion direction} throughout the rest of the paper and is one of the most critical markers for studying white matter fibers. The association between principal diffusion directions and subjects' covariates has emerged as a perspective in the study of neurodegenerative diseases \citep{schwartzman2005cross,schwartzman2008inference}. With studies revealing that tissue properties may vary systematically along each tract, i.e., tract profiles of white matter properties \citep{yeatman2012tract}, it becomes an essential scientific question in investigating how covariate effects drive the principal diffusion directions along each tract.

    \begin{figure}
    \centering
    \includegraphics[width=0.5\textwidth]{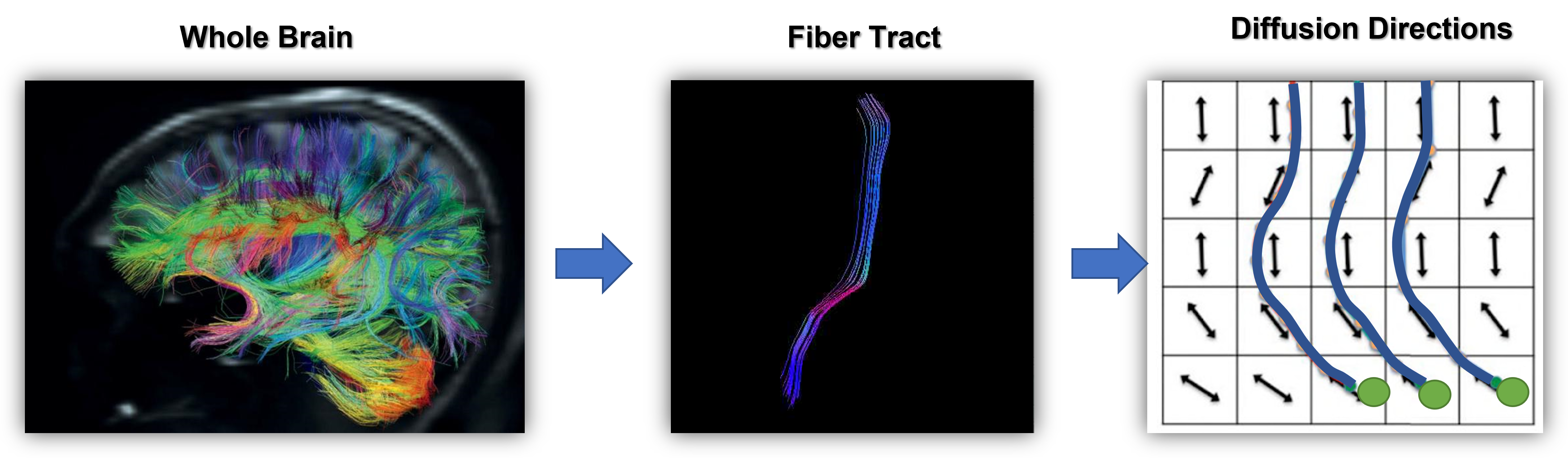}
    \caption{Macro to micro-view of the brain's anatomical structure: the left panel shows that the brain is connected through the streamlines of fiber tracts; the middle panel provides an example fiber tract; finally, the right panel shows the directional data characterizing the streamlines.}
    \label{fig:fiber}
\end{figure}

Traditional statistical methods, which rely on Euclidean space, are generally unsuitable for analyzing directional data due to their inability to account for the spherical relationships among values. To tackle these issues, directional statistics utilize a unique set of techniques \citep{gould1969regression} and probability distributions \citep{johnson1978some}. Several regression approaches have been developed for handling circular or 2-dimensional directional outcome data \citep{gould1969regression,johnson1978some,fisher1992regression,presnell1998projected}. However, these approaches are not generalizable for spherical outcomes where the data sits on three or higher-dimensional spheres \citep{mardia2009directional, paine2020spherical}. A regression method for analyzing such data was first introduced in \cite{chang1986spherical}, and subsequent developments include \citep{downs2003spherical, rosenthal2014spherical, scealy2019scaled, paine2020spherical}, and a few other approaches are proposed relying on the properties of a homogeneous Riemannian manifold \citep{hinkle2014intrinsic, cornea2017regression}. \cite{pewsey2021recent} provided a thorough review in this context. Nevertheless, none of these approaches is suitable for studying spatially varying directional outcomes.

In this paper, we develop a novel regression approach to infer the effects of several external factors on spatially varying directional data as an outcome. Notationally, let a $p$-variate directional data measured at the replication $i$ and location $v$, denoted as $\bm{E}_{iv}\in \mathcal{S}^{p-1}$, be the outcome of interest, and let a covariate vector $\bm{X}_{i}=[1, {X}_{i1}, \ldots, {X}_{ic}, \ldots, {X}_{iC}]$ represent the array of predictors. Here $\mathcal{S}^{p-1}$ stands for the $p$-dimensional sphere. Our objective is to regress $\bm{E}_{iv}$ on $\bm{X}_i$ for statistical inference purposes. In summary, we propose a generalized regression model for unit vector-valued responses, incorporating a novel link function based on the transformation between Cartesian and spherical coordinates. This link function enables us to project the directional data onto Euclidean space, thus allowing us to model the covariate effects efficiently.

As discussed in \cite{paine2020spherical}, the regression models for angular data or $\mathcal{S}^1$ are not easily extendable for $\mathcal{S}^2$ or higher-dimensional spheres. The fact that the intercept takes care of the rotational invariance as in \cite{fisher1992regression} is very specific to the angular data only. In our model, such modification would not make the estimate rotation invariant as the response is in $\mathcal{S}^2$. Following model Structure 2 of \cite{paine2020spherical} and Section 4.3.1 of \cite{scealy2019scaled}, we add an unknown rotation matrix $\bm{Q}$ as a parameter that allows the inference not to depend on the coordinate system. Unlike \cite{scealy2019scaled}, we put a prior and sample this orthogonal matrix within our MCMC since in a spatial setting, their approach for the pre-specification of this matrix may not be appropriate.
Thus, the posterior samples of the regression coefficients become dependent on the posterior samples of the orthogonal matrix. 
Hence, it is not straightforward to use the regression coefficients directly for \textit{scientific} inference. We thus consider a tangent-normal decomposition-based inference framework instead of using the regression coefficients directly. Specifically, we summarize the change in directional-valued mean upon a one-unit change in predictors in terms of the magnitude of the tangent-normal component and use those summarizations to quantify the importance of different predictors. Further details are provided in Section~\ref{sec:inference}. Moreover, the inclusion of this orthogonal matrix improves flexibility along the lines of \cite{kent2006complex, scealy2019scaled, paine2020spherical} and also the model performance, as shown in our simulation results in Section \ref{sec:num}.

Another important component to be addressed is the aspect of complex spatial dependence. In many modern applications, the directional data are spatially correlated, and the spatial dependence presents a different and unique pattern, i.e., the spatial dependence of this type of directional data is usually along \textit{streamlines} \citep{yeatman2012tract}. The \textit{streamlines} have different meanings in different applications. In animal migration, \textit{streamlines} are the animal migration pathways \citep{mahan1991circular}. In DTI, the \textit{streamlines} are the fiber tract \textit{streamlines}, i.e., the blue curves in the middle panel of Figure \ref{fig:fiber}. In \citet{goodlett2009group, zhu2011fadtts}, the authors proposed methods that induce spatial dependence by considering arc length distances along a fiber while analyzing the fiber orientation data. Such dependence can also be observed in our motivating DTI data (see Section A in the supplementary materials). Motivated by this, we impose the spatial dependence along the \textit{streamlines}, relying on an autoregressive model. For our Alzheimer's Disease Neuroimaging Initiative (ADNI) data analysis, the fiber tract \textit{streamlines}, based on the tractography atlases from \cite{yeh2018population} are used.
The tractography atlases provided by \citet{yeh2018population} offer inter-voxel connectivities along a given fiber tract, the fornix (see Figure \ref{fig:atlas}). These are estimated using high-quality diffusion-weighted images from healthy subjects. Fornix's tract profile has been studied in many other studies \citep{valdes2020diffusion,cahn2021diffusion}.

    \begin{figure}
    \centering
    \includegraphics[width=0.2\textwidth]{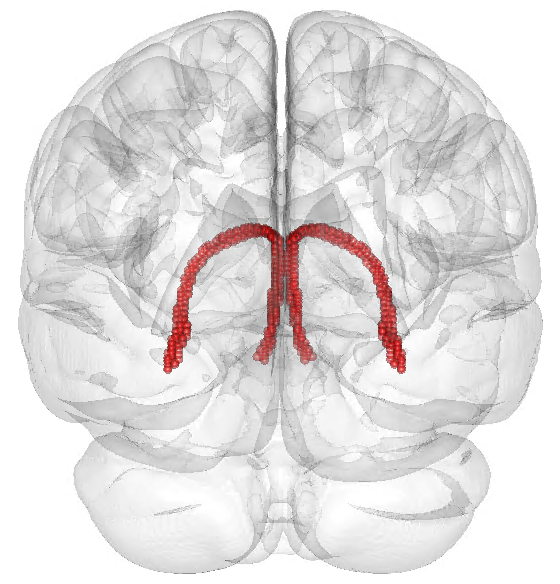}
    \caption{Tractography atlas of fornix .}
    \label{fig:atlas}
\end{figure}

In the rest of the paper, we first illustrate the model construction in Section \ref{sec:construction}. The proposed Bayesian angular inference is given in Section \ref{sec:inference}. To evaluate the performance of our model, we provide numerical studies in Section \ref{sec:num}. Finally, we conclude and discuss in Section \ref{sec:conclusion}. 

\section{Model Construction}
\label{sec:construction}
In this section, we provide the details of the model construction. Section \ref{sec:dist} describes the distribution that serves as the noise model for directional data. Subsequently, Section \ref{sec:link} introduces the link function, which allows us to relate a linear combination of the predictors to the response variables. Finally, we describe the model for spatial dependence in Section \ref{sec:sp_dep}.

\subsection{vMF Distribution}
\label{sec:dist}
A random $p$-dimensional directional data point, $\bm{E}_{iv}$, lies on $\mathcal{S}^{p-1}$, as described in Section \ref{sec:intro}. The von Mises-Fisher (vMF) distribution \citep{mardia2009directional} is a commonly used probability distribution for handling directional data, and spherical regression models are built based on the vMF distribution. However, building spherical regression models is not straightforward since the inference might depend on an arbitrary coordinate system in which the responses are defined. Following the approaches in \citet{scealy2019scaled,paine2020spherical}, We assume that $\bm{Q}^T\bm{E}_{iv}$ follows a von Mises-Fisher (vMF) distribution \citep{mardia2009directional}, denoted as $\bm{Q}^T\bm{E}_{iv}\sim\text{vMF}\big(\bm{\mu}_{iv},\kappa\big)$ with the probability density function expressed as
\begin{equation}
\label{eq:density}
\begin{aligned}
&f\Big[\bm{Q}^T\bm{E}_{iv}|\bm{\mu}_{iv}, \kappa\Big]=C_p(\kappa) \exp\Big(\kappa \bm{\mu}_{iv}^T\bm{Q}^T\bm{E}_{iv}\Big)=C_p(\kappa) \exp\Big(\kappa \mathcal{M}_{iv}^T\bm{E}_{iv}\Big),
\end{aligned}
\end{equation}
where $\mathcal{M}_{iv}=\bm{Q}\bm{\mu}_{iv}$. In the expression, $C_p(\kappa)$ is the normalizing constant such as $C_{p}(\kappa )={\frac {\kappa ^{p/2-1}}{(2\pi )^{p/2}I_{p/2-1}(\kappa )}}$, where $I_{a}$ stands for the modified Bessel function of the first kind at order $a$. In this paper, we primarily focus on the case with $p=3$, due to the nature of the DTI data, and in that case, $C_3(\kappa)=\frac{\kappa}{2\pi (e^{\kappa}-e^{-\kappa})}$ 
The term $\bm{\mu}_{iv}$ is a directional vector, i.e., $\|\bm{\mu}_{iv}\|=1$, where $\|\cdot\|$ represents the $\ell_2$ norm. $\kappa\geq 0$ is a concentration parameter. The term $\bm{Q}$ is set as an unknown parameter in our model and modeled as a special orthogonal matrix. 
The density is maximized at $\bm{E}_{iv}=\bm{Q}\bm{\mu}_{iv}$, the mode direction parameter. Thus, we further introduce the notation $\mathcal{M}_{iv}$
and it is the mode direction for the \textit{original} response $\bm{E}_{iv}$, and is named as \textit{direct} mode direction to differentiate it from $\bm{\mu}_{iv}$ which is called  \textit{rotated} mode direction.
 In the next subsection, we propose a link function to relate predictors with $\bm{\mu}_{iv}$. Since the link function is known and pre-specified, the orthogonal matrix $\bm{Q}$ here provides additional flexibility by identifying an optimal coordinate system with respect to the link.

\begin{figure}
    \centering

            \begin{subfigure}{0.4\textwidth}
                \centering
    \includegraphics[width=0.5\textwidth]{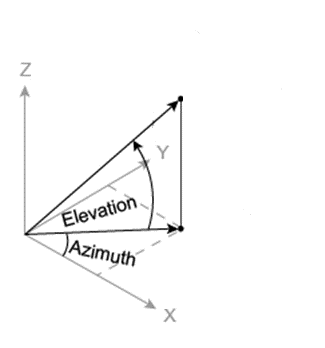}
    \caption{The graphical illustration describes the transformation between Cartesian and spherical coordinates. The azimuth angle and elevation angle are labeled in the Cartesian coordinate system.}
    \label{fig:coordinates}
     \end{subfigure}  
       \hfill  
     \begin{subfigure}{0.4\textwidth}
    \centering
    \includegraphics[width=0.5\textwidth]{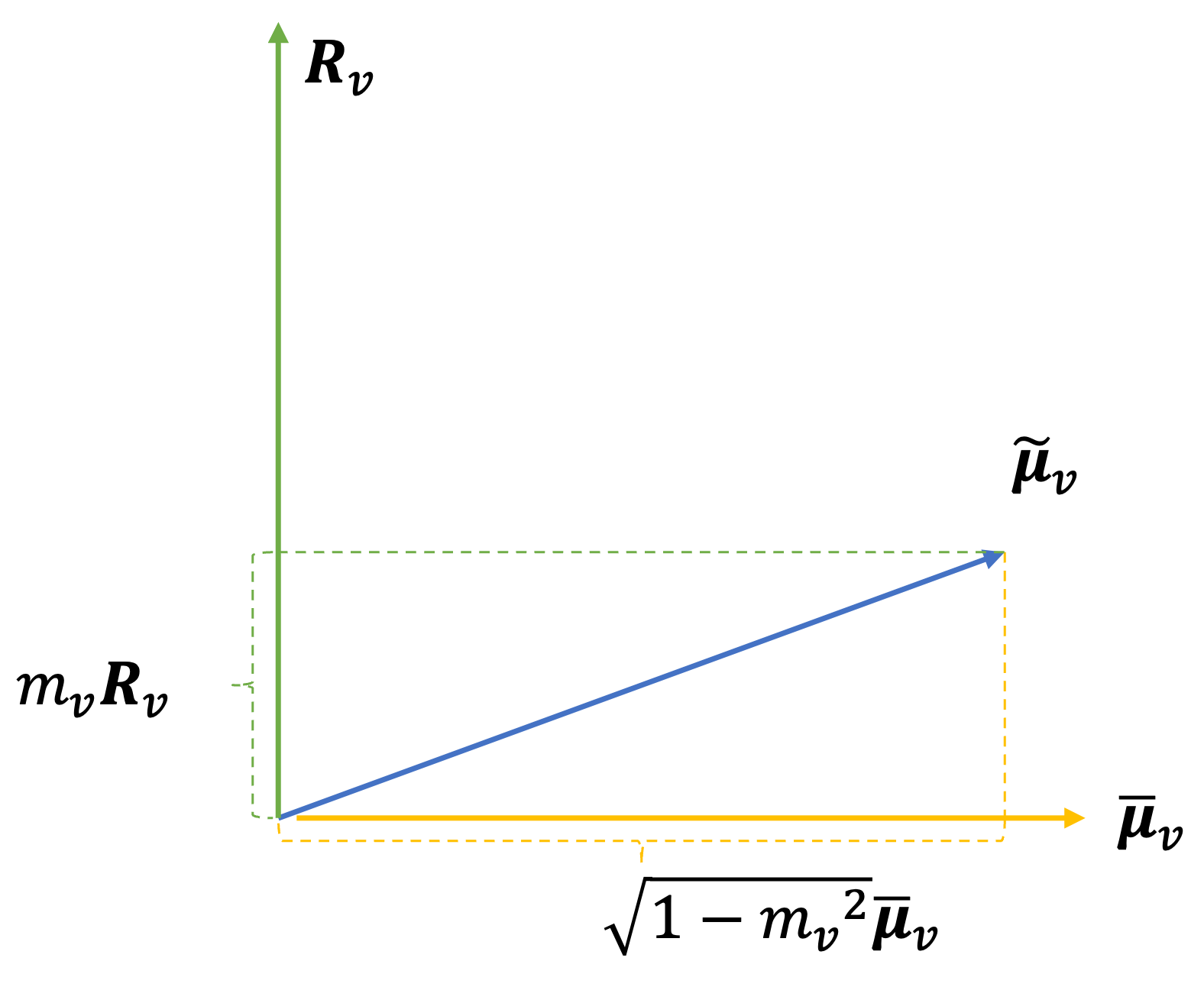}
    \caption{The graphical illustration of tangent-normal decomposition of covariate effects. The yellow arrow is the \textit{typical} mode direction; the blue arrow is the \textit{specific} mode direction; the green arrow is the \textit{deviation} direction; The scalar $m_{v}\in(0,1)$ controls the magnitude of this effect.}
    \label{fig:Tangent_Normal}
    \end{subfigure}  
    \caption{The graphical illustrations of directional statistics.}
\end{figure}

\subsection{Linking to Predictors}
\label{sec:link}

In accordance with the generalized linear regression (GLM) framework, we propose a relationship between a linear combination of the predictors and the response variables through a link function. 
It is also reasonable to allow the linear combination of the predictors to be unbounded, mirroring other GLM frameworks. Therefore, our proposed link function is a composition of two functions, each handling a different level of transformations. The following section provides detailed insight into the construction of the link function.

In this section, we focus on our link function for 3-dimensional spherical coordinates, i.e., when $\bm{E}_{iv}$ lies on $\mathcal{S}^{p-1}$ and $p=3$. Let $(x,y,z)$ represent the Cartesian coordinates of a point on the sphere such that $x^2+y^2+z^2=1$. This point can be represented in spherical coordinates using two parameters, as depicted in Figure \ref{fig:coordinates}. The azimuth angle $\theta\in (-\pi,\pi]$ represents the counterclockwise angle in the x-y plane from the positive x-axis, and the elevation angle $\phi\in (-\frac{\pi}{2},\frac{\pi}{2}]$ is the angle from the x-y plane.
The one-to-one transformation from Cartesian to spherical coordinates is as follows:
$$u:(x,y,z)\rightarrow \Big(\text{atan2}(y,x),\text{atan2}(z,\sqrt{1-z^2})\Big)=(\theta,\phi),$$
where $\text{atan2}(y,x)$ is defined by the limit expression $\lim_{z\rightarrow x^+}(\frac{y}{z})+\frac{\pi}{2}\text{sgn}(y)\text{sgn}(x)(\text{sgn}(x)-1)$.
While the spherical coordinates $(\theta,\phi)$ are located in Euclidean space, they are bounded. Hence, we first scale the spherical coordinates $(\theta,\phi)$ to fall within the interval $(0,1]$, then apply the logit function. The transformation is:
$$g:(\theta,\phi)\rightarrow \Big(\text{logit}\big(\frac{\theta+\pi}{2\pi}\big),\text{logit}\big(\frac{\phi+\pi/2}{\pi}\big)\Big)=(\widetilde{\theta},\widetilde{\phi})\in \mathbb{R}\times\mathbb{R}.$$ 
Here, $\text{logit}(x)=\log\left(\frac{x}{1-x}\right)$. Although the logit transformation ill-behaves on the boundary (i.e., at $x=1$), the regression model built on this transformation is still adequately flexible, as the set of such degenerate cases has measure zero. This transformation serves to map the bounded spherical coordinates to a new set of parameters $(\widetilde{\theta},\widetilde{\phi})$, which are unbounded.
For clarity, the overall mapping from $(x,y,z)$ to $(\widetilde{\theta},\widetilde{\phi})$ is denoted by $\ell(\cdot)$, which is a composite of $g(\cdot)$ and $u(\cdot)$, i.e., $\ell(\cdot):=u(\cdot)\circ g(\cdot)$. Thus, $\ell(x,y,z)=[\widetilde{\theta},\widetilde{\phi}]$ represents this innovative function that projects the directional data from $\mathcal{S}^{2}$ to $\mathbb{R}\times\mathbb{R}$.

A point on a $p$ dimensional hypersphere, $\mathcal{S}^{p-1}$, can be expressed using $p-1$ angles. The first $p-2$ angles (denoted as ${\bm{\phi}}^{(p-2)}=[{{\phi}}_1^{(p-2)}, \ldots, {{\phi}}_{p-2}^{(p-2)}]$) are within the interval $(-\pi/2,\pi/2]$, while the last angle ${\theta}$ is within $(-\pi,\pi]$. 
The mapping function, when extended to the general dimension, is formulated as: 
\begin{equation*}
    \begin{aligned}
        &g:(\theta,\bm{\phi}^{(p-2)})\rightarrow \Big(\text{logit}\big(\frac{\theta+\pi}{2\pi}\big),\big[\text{logit}\big(\frac{{\phi}_1^{(p-2)}+\pi/2}{\pi}\big),\ldots,\text{logit}\big(\frac{{\phi}_{p-1}^{(p-2)}+\pi/2}{\pi}\big)\big]\Big)\\
        &=(\widetilde{\theta},\big[\widetilde{{\phi}}_1^{(p-2)},\ldots,\widetilde{{\phi}}_{p-2}^{(p-2)}\big])\in \mathbb{R}\times\mathbb{R}^{p-2}.
    \end{aligned}
\end{equation*}

In our proposed generalized regression model, the link function $\ell(\cdot)$ plays a key role. The term $\ell(\bm{\mu}_{iv})=(\widetilde{\theta},\big[\widetilde{{\phi}}_1^{(p-2)},\ldots,\widetilde{{\phi}}_{p-2}^{(p-2)}\big])$ is considered as the prediction terms and is formulated as a function of the covariate vector $\bm{X}_{i}$. Specifically, we have:
$$\mathbb{E}(\widetilde{\theta},\big[\widetilde{{\phi}}_1^{(p-2)},\ldots,\widetilde{{\phi}}_{p-2}^{(p-2)}\big])=(\bm{X}_{i}\bm{\alpha}_{v}, [\bm{X}_{i}\bm{\beta}^{(p-2)}_{v,1},\ldots,\bm{X}_{i}\bm{\beta}^{(p-2)}_{v,p-2}]).$$
Lemma \ref{lemma:bi} illustrates that the link function $\ell(\cdot)$ is a bijective function. This provides us with a direct mapping between the coordinates $(\widetilde{\theta},\big[\widetilde{{\phi}}_1^{(p-2)},\ldots,\widetilde{{\phi}}_{p-2}^{(p-2)}\big])$ and the mode direction $\bm{\mu}_{iv}$. This mapping is crucial in facilitating the interpretation of the regression coefficients.
\begin{lemma}
\label{lemma:bi}
The function $\ell(\cdot): \mathcal{S}^{p-1}\rightarrow\mathbb{R}^{p-2}\times\mathbb{R}$ is bijective.
\end{lemma}

\begin{proof}
We know $\ell(\cdot):=u(\cdot)\circ g(\cdot)$ where $\circ$ is function composition. Since $u(\cdot)$ and $g(\cdot)$ are both bijective, $\ell(\cdot)$ is bijective.
\end{proof}
For the remainder of this article, in view of our DTI application and to simplify the exposition, we set $p=3$, 
i.e., $\mathbb{E}[\widetilde{\theta}_{iv},\widetilde{\phi}_{iv}]^T=[\bm{X}_{i}\bm{\alpha}_{v}, \bm{X}_{i}\bm{\beta}_{v}]^T,$ and drop the superscript $(p-2)$. Despite this simplification, it should be noted that it is straightforward to generalize all the steps we take here to accommodate any general $p$. 

There are several other link functions proposed in spherical-real regression literature. However, the common theme is always to link the individual-specific mean with a linear combination of predictors, as in any generalized linear model. Here, we review them for completeness. In each case, we mention the possible issues for clarity, not as a form of criticism. Due to the complexity of regressing a spherical response on real predictors, the proposed link functions do have some mathematical restrictions. However, they are all shown to work well empirically. In \cite{scealy2017directional} and \cite{scealy2019scaled}, the mean unit vector is first transformed to simplex-valued data by taking an element-wise square. Subsequently, a link to the predictor space is constructed through a soft-max transformation. Hence, this only works when the means lie on the positive orthant. Although given a fixed voxel, all the angles for all the subjects are within a very small range, they are not all on the positive orthant in our DTI application. 

There are two alternative links proposed in \cite{paine2020spherical}. However, the link in structure 1 described in \cite{paine2020spherical} is designed specifically for $\mathcal{S}^2$-valued response following a Kent or ESAG distribution that includes three unit vectors as parameters along with a few additional parameters. The three unit vectors in this case form an orthogonal matrix, and a Cayley transformation is used to map it to the unrestricted real space to develop the regression model. Mathematically, it is still restrictive, as the mapping can often non-trivially lead to large separations in the underlying skew-symmetric matrices even when the cross-product of the two resulting special orthogonal matrices is almost close to identity. Structure 2 described in \cite{paine2020spherical} can be applicable within our proposed vMF distribution modeling as well, specifying $\bm{\mu}_{iv}$ as $\frac{\bm{B}_v\mathbf{X}_i}{\|\bm{B}_v\mathbf{X}_i\|}$ where $\bm{B}_v$ is a $3\times (C+1)$ dimensional coefficient matrix.
Mathematically, the unit norm projection can cause two very different vectors, $\bm{B}_v\mathbf{x}_i$ and $\bm{B}_v\mathbf{x}_j$, to appear very close to each other on a sphere after the unit-norm projection, when either or both of the two norms $\|\bm{B}_v\mathbf{x}_i\|$ and $\|\bm{B}_v\mathbf{x}_j\|$ are more than 1. Conversely, it can also do the opposite when their norms are both less than 1.
The unit norm projection (i.e., structure 2 link) is claimed to have performed better than the first in \cite{paine2020spherical}. 
In summary, while these links exhibit certain mathematical limitations, they have all been shown to be effective in practical applications. Specifically, as long as the unit vectors do not exhibit significant variation, all these links work. In practice, such variations are generally minimal.

In the supplementary Section B, we provide a comparison to the above unit-norm projection link and the corresponding implementation details. We find our proposed link performed better than the unit norm projection-based link.

\subsection{Spatial Dependence}
\label{sec:sp_dep}
Spatial dependence is an essential component in many applications. In this section, we further illustrate the model specification, which accounts for spatial dependence. 
By preserving $\mathbb{E}[\widetilde{\theta}_{iv},\widetilde{\phi}_{iv}]^T=[\bm{X}_{i}\bm{\alpha}_{v}, \bm{X}_{i}\bm{\beta}_{v}]^T$, we let
$[\widetilde{\theta}_{iv},\widetilde{\phi}_{iv}]^T=[\bm{X}_{i}\bm{\alpha}_{v}+\epsilon_{iv}, \bm{X}_{i}\bm{\beta}_{v}+\xi_{iv}]^T$, where $\bm{\epsilon}_{i}=(\epsilon_{1i}, \ldots, \epsilon_{iv}, \ldots, \epsilon_{iV})$ and $\bm{\xi}_{i}=(\xi_{1i}, \ldots, \xi_{iv}, \ldots, \xi_{iV})$ are distributed as the mean-zero Gaussian distribution with covariance matrices $\bm{\Sigma}^{(\epsilon)}$ and $\bm{\Sigma}^{(\xi)}$, respectively. 

The focus of our applications is principal directional data, where the spatial dependence is along a given fiber streamline. Suppose we have $K$ fiber streamlines, $\mathbb{Z}_k$ is a set of consecutive indices within the $k$-th fiber streamline. We let the unit vectors along a given streamline be sequentially autocorrelated.
For completeness, we first state the definition and notations related to the Gaussian autoregressive modeling in Definition \ref{def:ar}.
\begin{definition}[AR($P$) Process]
\label{def:ar}
Given $A_t=A_{t-1}\Phi_1+\ldots+A_{t-P}\Phi_P+\epsilon_t$ and $\epsilon_t\sim\mathcal{N}(0,\sigma^2)$, for $t\in\mathbb{Z}$. For $t\in\mathbb{Z}$, we define that $\{A_t:t\in\mathbb{Z}\}$ follows a Gaussian AR-$P$ process with correlation parameters ($\Phi_1, \ldots, \Phi_P$) and variance $\sigma^2$, denoted as $\text{AR}_P\Big([\Phi_1,\ldots,\Phi_P],\sigma^2\Big)$.
\end{definition}
We thus have $\bm{\epsilon}_{i}=(\epsilon_{1i}, \ldots, \epsilon_{iv}, \ldots, \epsilon_{iV})$ and $\bm{\xi}_{i}=(\xi_{1i}, \ldots, \xi_{iv}, \ldots, \xi_{iV})$ follow a multivariate normal distribution with an autoregressive covariance matrix, denoted as
\begin{equation}
    \begin{aligned}
    \bm{\epsilon}_{ik}=\{\epsilon_{iv}: v\in \mathbb{Z}_k\}\sim\text{AR}_P\Big([\Phi_{\epsilon 1},\ldots,\Phi_{\epsilon P}],\tau_{\epsilon}^2\Big)\\
    \bm{\xi}_{ik}=\{\xi_{iv}: v\in \mathbb{Z}_k\}\sim\text{AR}_P\Big([\Phi_{\xi 1},\ldots,\Phi_{\xi P}],\tau_{\xi}^2\Big),
    \end{aligned}
\end{equation}
where $(\Phi_{\epsilon,1},\ldots,\Phi_{\epsilon,P},\tau_{\epsilon}^2)$ and $(\Phi_{\xi,1},\ldots,\Phi_{\xi,P},\tau_{\xi}^2)$ are the corresponding parameters. Note that our characterization of the AR$(P)$ process is based on the partial correlations but not the autoregressive coefficients directly.

\subsection{Prior Settings}
\label{sec:prior}
We now put priors on the model parameters to proceed with our Bayesian inference. The matrix $\bm{Q}$ is first reparametrized using the Cayley transformation as $\bm{Q}=(\bm{I}-\bm{A})(\bm{I}+\bm{A})^{-1}$, where $\bm{A}$ is a skew-symmetric matrix. This parametrization is useful for efficient posterior computation. Setting $\bm{A}=\begin{bmatrix}
   0 & a_1 & a_2  \\
    -a_1 & 0 & a_3\\
    -a_2 & -a_3 & 0
\end{bmatrix}$, the priors for the three parameters $a_1,a_3,a_3$ are set as mean-zero normal with variance $100$. For the regression coefficients, we also assign AR-$P$ process priors, denoted as
\begin{equation}
    \begin{aligned}
        \bm{\alpha}_{k}(c)=\{\alpha_{v}(c): v\in \mathbb{Z}_k\}\sim\text{AR}_P\Big(
        [\Phi_{\alpha,1},\ldots,\Phi_{\alpha,P}],\sigma_\alpha^2\Big)\\
     \bm{\beta}_{k}(c)=\{\beta_{v}(c): v\in \mathbb{Z}_k\}\sim\text{AR}_P\Big([\Phi_{\beta,1},\ldots,\Phi_{\beta,P}],\sigma_\beta^2\Big),
    \end{aligned}
\end{equation}
where $\bm{\alpha}_{v}=\Big({\alpha}_{v}(0), \ldots, {\alpha}_{v}(c), \ldots, {\alpha}_{v}(C)\Big)$ and $\bm{\beta}_{v}=\Big({\beta}_{v}(0), \ldots, {\beta}_{v}(c), \ldots, {\beta}_{v}(C)\Big)$. The terms ${\alpha}_{v}(0)$ and ${\beta}_{v}(0)$ are the coefficients for the intercepts. The terms ${\alpha}_{v}(c)$ and ${\beta}_{v}(c)$ are the coefficients for the covariate $c$.
The variances are assigned with weakly informative inverse-gamma priors with shape and rate parameters set to 0.1, denoted as $\tau_{\epsilon}^{-2}, \tau_{\xi}^{-2}, \sigma_{\alpha}^{-2}, \sigma_{\beta}^{-2}\sim\mathcal{GA}(0.1,0.1)$. We put a diffuse prior for the concentration parameter, denoted as $\kappa\propto 1$. The concentration parameter $\kappa$ can be viewed as a nugget effect. 
An autoregressive process with the partial autocorrelation parameters restricted to $(-1,1)$ can ensure stationarity \citep{barnett1996bayesian}. 
Therefore, we assign a uniform $(-1,1)$ prior on the partial autocorrelation parameters. 
Specifically, $\rho_{\alpha,p}, \rho_{\beta,p},\rho_{\epsilon,p}, \rho_{\xi,p}\sim\mathcal{U}(-1,1)$ for $p\in\{1,2,\ldots,P\}$, where $\rho_{\alpha,p}, \rho_{\beta,p},\rho_{\epsilon,p}, \rho_{\xi,p}$ are the corresponding partial autocorrelation parameters associated with the autoregressive process priors for $\bm{\alpha}_{v}(c)$'s, $\bm{\beta}_{v}(c)$'s, and the autoregressive random effects $\bm{\epsilon}_{ik}$'s, and $\bm{\xi}_{ik}$'s, respectively. 
Applying the Durbin-Levinson recursion on the partial autocorrelations ($\rho_{\alpha,p}, \rho_{\beta,p},\rho_{\epsilon,p}, \rho_{\xi,p}$), we get back the autoregressive coefficients ($\Phi_{\alpha,p}, \Phi_{\beta,p},\Phi_{\epsilon,p}, \Phi_{\xi,p}$) using the \code{R} function \code{inla.ar.pacf2phi()}  from the \code{R} package \code{inla} \citep{lindgren2015bayesian}. The reverse transformation is applied using \code{inla.ar.phi2pacf()} from the same package.
We utilize a Markov chain Monte Carlo (MCMC) algorithm, which includes Metropolis-Hastings-within-Gibbs steps, to generate the posterior samples needed for our Bayesian inference.

\section{Bayesian Angular Inference}
\label{sec:inference}
In this section, we present our novel Bayesian angular inference framework for directional data analysis. The interpretation of the regression coefficients in our model is not as straightforward as it is in linear regression. The proposed inferential framework aims to address this complexity in inference. The main difficulty arises when trying to illustrate the effects of covariates.

\subsection{Posterior Predictive Distribution of \textit{Direct} Mode Direction}
\label{sec:AE}

Inferring the posterior predictive distribution of \textit{direct} mode direction is the main goal in Bayesian angular inference. Given a subject with the covariate vector $\bm{X}_{0}$, the posterior predictive distribution of the \textit{direct} mode direction $\bm{\mu}_{0,v}$ for the $v$-th index, is $\Big[\mathcal{M}_{0,v}|\bm{X}_{0};\text{data}\Big]$. The MCMC samples $\{\mathcal{M}_{0,v}^{(t)}:t=1:T\}$ are the realizations of the posterior distribution. Following \citet[][Section 15.2]{mardia2009directional}, we use the $\frac{\sum_{t=1}^T\mathcal{M}_{\text{new},v}^{(t)}/T}{\|\sum_{t=1}^T\mathcal{M}_{\text{new},v}^{(t)}/T\|}$ as the posterior summary for the \textit{direct} mode direction $\bm{\mu}_{0,v}$. Relating to the real applications, the above term profiles the expected fiber direction after the model-based adjustment for the covariate vector $\bm{X}_{0}$. The posterior summarization can be used for the subsequent tangent-normal decomposition for covariate effects.

\subsection{Tangent-Normal Decomposition for Covariate Effects}
\label{sec:TN}
A transparent inference framework illustrating the covariate effects is essential for any scientific study.
In our proposed model, the coefficients $\bm{\alpha}_{v}$ and $\bm{\beta}_{v}$ do not offer a straightforward illustration of the covariate effects on the {mode} directions. 
This is also a limitation of many regression models for complicated responses.
To overcome this, we incorporate the tangent-normal decomposition \citep[][Page 169]{mardia2009directional} to summarize the covariate effects on directional data as follows. 
The tangent-normal decomposition is an integrated approach, but can be decomposed into the following three steps for a concise illustration.

Our proposed procedure starts by defining the following two covariate vectors: a) $\bm{X}_{\text{T}}=[1, {X}_{\text{T}1},\ldots ,{X}_{\text{T}c}, \ldots, {X}_{\text{T}C}]$: Covariate vector for the \textit{typical} condition; b) ${\bm{X}}_{\text{S}}=[1, {X}_{\text{S}1},\ldots ,{X}_{\text{S}c}, \ldots, {X}_{\text{S}C}]$: Covariate vector for the \textit{specific} condition. Thus, we replace the original subscript ``$i$" with ``T" and ``S" in our notation, to represent a subject in the \textit{typical} and \textit{specific} condition, respectively. 

Next, we obtain the corresponding predictive posterior distributions of the {mode} directions with the given \textit{typical} or \textit{specific} covariate vector. They are referred to as the \textit{typical} and \textit{specific} mode directions, respectively, expressed as follows: a) {\textit{Typical} Mode Direction:} $\overline{\mathcal{M}}_{v}=\Big[\mathcal{M}_{\text{T}v}|{\bm{X}}_{\text{T}};\text{data}\Big]$; b) \text{\textit{Specific} Mode Direction:} $\widetilde{\mathcal{M}}_{v}=\Big[\mathcal{M}_{\text{S}v}|{\bm{X}}_{\text{H}};\text{data}\Big]$. As stated in Section \ref{sec:AE}, they can also be approximately learned from the MCMC outputs using the introduced posterior summary.

To quantify the covariate effect, we employ the tangent-normal decomposition on the \textit{specific} mode direction $\widetilde{\mathcal{M}}_{v}$ as in Figure \ref{fig:Tangent_Normal} where the \textit{typical} mode direction $\overline{\mathcal{M}}_{v}$ is the tangent direction and the term $\bm{R}_{v}$ is the normal direction. 
Applying this to the principal diffusion directions, the term $\bm{R}_{v}$ can be interpreted as the \textit{deviation} direction that is caused by the difference from the \textit{specific} condition to the \textit{typical} condition.
The scalar $m_{v}\in(0,1)$, which controls the magnitude of this effect, can be used to quantify the importance of this covariate: the larger value of $m_{v}$ indicates the stronger deviation between the \textit{specific} condition to the \textit{typical} condition. 
Computationally, the tangent-normal decomposition can be applied to each MCMC sample $t$. This provides us with the posterior samples of $\bm{R}_{v}$ and $m_{v}$, denoted as $\{\bm{R}_{v}^{(t)}:t=1:T\}$ and $\{m_{v}^{(t)}:t=1:T\}$.
We use those samples to obtain a posterior summary of $[\bm{R}_{v}|\text{data}]$ to spatially profile the \textit{deviation} directions caused by the difference between \textit{typical} condition and \textit{specific} condition. 
Similarly, the posterior mean estimate of $m_{v}$ can also be computed and used to quantify the magnitude. We use the following notation to express the tangent-normal decomposition: $\mathcal{D}_{TN}(\overline{\mathcal{M}}_{v}||\widetilde{\mathcal{M}}_{v})=\{\bm{R}_{v},m_{v}\}$.

Our above illustration provides a general description of tangent-normal decomposition. As long as it is feasible to specify the \textit{typical} and \textit{specific} conditions, they may use $\{\bm{R}_{v},m_{v}\}$ to infer the covariate effects caused by the difference between the \textit{typical} and \textit{specific} conditions. The following are two important remarks that are important for the general use of our approach.

\begin{remark}
Two covariate vectors can be any two conditions where we want to compare the \textit{specific} condition ``to" the \textit{typical} condition which is considered as reference. Usually, ${\bm{X}}_{\text{T}}$ is set to the common characterization.
\end{remark}

\begin{remark}
$\mathcal{D}_{TN}(\overline{\mathcal{M}}_{v}||\widetilde{\mathcal{M}}_{v})$ is not symmetric in the two inputs, i.e., $\mathcal{D}_{TN}(\overline{\mathcal{M}}_{v}||\widetilde{\mathcal{M}}_{v})\not\equiv\mathcal{D}_{TN}(\widetilde{\bm{\mu}}_{v}||\overline{\mathcal{M}}_{v})$.
Thus, the \textit{specific} condition and the \textit{typical} condition cannot be exchanged.
\end{remark}

While we illustrate the use of our approach as a generic use, interpreting one covariate effect at a time is usually a common manner. Therefore, we further provide two examples below to illustrate how to interpret one covariate effect at a time. In summary, the key strategy is to have \textit{typical} and \textit{specific} conditions differ only in the covariate $c'$ of our primary interest.

\begin{example}[Continuous Covariate]
Let $c'$ represent a continuous covariate. The terms $X_{\text{T}c'}$ and $X_{\text{S}c'}$ correspond to the values in the \textit{typical} and the \textit{specific} cases, respectively. When investigating the impact of a single unit increase in the covariate $c'$, with $X_{\text{T}c'}=a$, we assign $X_{\text{S}c'}=a+1$. Concurrently, we ensure $X_{\text{S}c}=X_{\text{T}c}$ for all other $c$ variables.
\end{example}

\begin{example}[Categorical Covariate]
Let $c'$ represent a categorical covariate with $M$ levels, where $m\in{1,\ldots,M}$. The terms $\{X_{\text{T}c',m}\}\in\{0,1\}$ and $\{X_{\text{S}c',m}\}\in\{0,1\}$ denote the corresponding specification for level $m$. Specifically, $X_{\text{T}c',j}=1$ and $X_{\text{S}c',j'}=1$ if the \textit{typical} case is categorized under the $j$-th level and the \textit{specific} condition falls under the $j'$-th level. This is with the condition that $X_{\text{T}c',m}=0$ for all $m\neq j$, and $X_{\text{S}c',m}=0$ for all $k\neq j'$.
In assessing the impact of the $j'$-th level starting from a typical case in the $j$-th level, we set $X_{\text{S}c',j'}=1$ and $X_{\text{T}c',j}=1$, while ensuring $X_{\text{S}c}=X_{\text{T}c}$ for all the other $c\in\{0:C\}$ variables.
\end{example}

\section{Numerical Studies}
\label{sec:num}
In this section, we aim to demonstrate the efficacy of our model compared to traditional methods. Initially, we evaluate our model fitness on synthetic data in Section \ref{sec:comparison}. We generate the synthetic data to mimic the real ADNI data best.

The ADNI initiative was inaugurated in 2003 as a public-private partnership with the primary goal of assessing whether the amalgamation of serial magnetic resonance imaging, positron emission tomography, other biological markers, and clinical as well as neuropsychological assessments could aid in predicting the progression of Alzheimer's disease. This paper specifically focuses on ADNI-2, a project that commenced in September 2011 and spanned five years \citep{aisen2010clinical}.
The Centers for Disease Control and Prevention (\url{https://www.cdc.gov/aging/aginginfo/alzheimers.htm}) suggests that Alzheimer's disease symptoms typically occur after the age of 60. Therefore, we have tailored our study cohort to focus on subjects over 60 years old. The subjects are further classified into four disease states: healthy controls (CN), early mild cognitive impairment (EMCI), late mild cognitive impairment (LMCI), and Alzheimer's Disease (AD). We use $g\in\{\text{CN},\text{EMIC},\text{LMCI},\text{AD}\}$ to denote the group index.

Additionally, we collect subject-level information including age ($X_{i,age}$), gender ($X_{i,gender}$), mini-mental state examination (MMSE) score ($X_{i,MMSE}$), and the Apolipoprotein E (APOE) profile ($X_{i,APOE}$). The MMSE is a performance-based neuropsychological test score ranging from $0$ to $30$. Subjects with cognitive impairments are typically associated with a lower MMSE score.
The APOE is a genetic marker with three major variants ($\epsilon$-2, $\epsilon$-3, and $\epsilon$-4). The $\epsilon$-4 variant is considered the most significant known genetic risk factor for Alzheimer's disease across various ethnic groups \citep{sadigh2012association}. Consequently, our analysis includes APOE-4 as a binary indicator, which denotes the presence of the $\epsilon$-4 variant in subjects.
We also incorporate age and gender into our analysis, given that previous studies have established their correlation with AD \citep{podcasy2022considering}.
In Section \ref{sec:comparison}, we provide the model comparison based on both synthetic data and real ADNI data. In Section \ref{sec:analysis}, we provide the tangent-normal decomposition within Bayesian angular inference to infer the covariate effects. In both sections, we make the design matrix as $  \bm{X}_i=[1, X_{i,age},X_{i,gender}, X_{i,MMSE}, X_{i,APOE}]$.
We further specify the regression terms as $\mathbb{E}\ell\Big(\bm{\mu}_{iv}\Big)=\mathbb{E}[\widetilde{\theta}_{iv},\widetilde{\phi}_{iv}]^T=[\bm{X}_i \bm{\alpha}_{gv}, \bm{X}_i \bm{\beta}_{gv}]^T$ for subject $i$ in clinical group $g$.
In this way, the covariate effects are clinical group-specific.

\subsection{Model Comparison}
\label{sec:comparison}
In this section, we carry out numerical analyses on synthetic principal diffusion direction data. The purpose of these studies is to demonstrate the superior performance of our proposed vMF regression relative to other methods. Here, we consider three competing methods based on multivariate Gaussian distribution: 1) \textit{Gaussian Regression 1:} a multivariate regression model for the principal diffusion directions $\bm{E}_{iv}$, 2) \textit{Gaussian Regression 2:} a multivariate regression model for transformed means $\ell\Big(\bm{E}_{iv}\Big)$, 3) \textit{Non-spatial vMF regression:} $\bm{Q}^T\bm{E}_{iv}\sim\text{vMF}\big(\mathcal{M}_{iv},\kappa\big)$ and $\ell\Big(\bm{\mu}_{iv}\Big)=[\widetilde{\theta}_{iv},\widetilde{\phi}_{iv}]^T=[\bm{X}_i \bm{\alpha}_{gv}, \bm{X}_i \bm{\beta}_{gv}]^T$. Specific details are described as follows.
The \textit{Gaussian Regression 1} simply treats the principal diffusion directions $\bm{E}_{iv}$ as normally distributed random variables, such that $\bm{E}_{iv}\sim\mathcal{N}(\mathcal{M}_{iv},\bm{\Sigma}_1)$, where $\mathcal{M}_{iv}^{(p)}=\bm{X}_{i}\bm{u}_{v}^{(p)}$, $\mathcal{M}_{iv}^{(p)}$ is the $p$-th element of $\mathcal{M}_{iv}$ for $p\in\{1,2,3\}$, i.e., $\mathcal{M}_{iv}=(\mathcal{M}_{iv}^{(1)},\mathcal{M}_{iv}^{(2)},\mathcal{M}_{iv}^{(3)})$. 
Alternatively, in case of \textit{Gaussian Regression 2:}, the terms $\ell\Big(\bm{E}_{iv}\Big)=[A_{iv}, B_{iv}]^T$ are assumed to follow a normal distribution, denoted as $\ell\Big(\bm{E}_{iv}\Big)=[A_{iv}, B_{iv}]^T\sim\mathcal{N}([\tilde{\theta}_{iv},\tilde{\phi}_{iv}]^T,\bm{\Sigma}_2)$, where $\tilde{\theta}_{iv}=\bm{X}_{i}\bm{a}_{v}$ and  $\tilde{\phi}_{iv}=\bm{X}_{i}\bm{b}_{v}$. In case of \textit{Non-spatial vMF Regression:}, we simply remove the spatially dependent randomness from the proposed model.

We take the Bayesian route for inference for all these methods and put weakly informative priors on the unknown parameters. Specifically, we give normal priors with mean $\bm{0}$ and covariance matrix $1000\times\bm{I}$ on the coefficients, denoted as $\bm{u}_{v}^{(1)}, \bm{u}_{v}^{(2)}, \bm{u}_{v}^{(3)},\bm{a}_{v}, \bm{b}_{v}\sim\mathcal{N}(\bm{0},1000\times\bm{I})$. We further put inverse Wishart prior with degrees of freedom $5$ and scale matrix $\bm{I}$ on the covariance matrices, denoted as $\bm{\Sigma}_1^{-1},\bm{\Sigma}_2^{-1}\sim\mathcal{W}(5,\bm{I})$. Our inference is based on the posterior samples, generated by applying the MCMC algorithm, where we collect $3,000$ post-burn samples after discarding the first $2,000$. The MCMC convergence is monitored by Heidelberger and Welch’s convergence diagnostic \citep{heidelberger1981spectral}.

We generate synthetic principal diffusion direction data using our proposed vMF regression model. We design a synthetic fiber tract that includes $K=12$ fibers and a total of $1,256$ voxels. We also have covariate vectors for 20 subjects for training and 20 subjects for validation. For each data replication, we fit our model to the training data to obtain the parameter estimates, and these parameter estimates are used to predict the response in the validation data. The discrepancies between the predicted responses and true responses are used to measure model accuracy.

The tractography atlas, voxel-wise coefficients, and covariate vectors of training/validation subjects are stored in the R file \code{synthetic.Rdata} in the code repository as \code{Fiber$\_$Ind}, \code{COEF}, and \code{Training}/\code{Validation}, respectively. For other parameters, we set $P=5$ and generate partial correlation coefficients from $\mathcal{U}(-1,1)$. The variances are specified as $\tau_{\epsilon}=\tau_{\xi}=\sigma_{\alpha}=\sigma_{\beta}=0.5$, and orthogonal matrix is specified as $\bm{Q}=\bm{I}$. We further give $\kappa\in\{20,30,40\}$, respectively. A total of 50 replicated data sets for each setting. For more information, please check the code repository (\code{Spatial$\_$VMF$\_$Regression.zip}). The code for implementing our proposed method is available on GitHub: \url{https://github.com/lanzhouBWH/Spatial_VMF_Regression}.

Prediction errors are quantified in terms of a) the separation angle, and b) the root mean squared error. The separation angle is defined as $\sum_{i,v}\delta\Big(\widehat{\mathcal{M}}_{iv},\hat{\bm{Q}}^T\bm{E}_{iv}\Big)/(VN)$ for vMF regression and $\sum_{i,v}\delta\Big(\widehat{\mathcal{M}}_{iv},\bm{E}_{iv}\Big)/(VN)$ for Gaussian regression methods, where $\delta(\bm{a},\bm{b})$ measures the separation angle between vector $\bm{a}$ and vector $\bm{b}$. The root mean squared error is defined as $\sum_{i,v}\|\widehat{\mathcal{M}}_{iv}-\hat{\bm{Q}}^T\bm{E}_{iv}\|/(VN)$ for vMF regression and $\sum_{i,v}\|\widehat{\mathcal{M}}_{iv}/\|\widehat{\mathcal{M}}_{iv}\|-\bm{E}_{iv}\|/(VN)$ for Gaussian case.
 We give the prediction errors spanning simulation replications in Figure \ref{fig:result}. Our proposed method substantially outperforms Gaussian alternatives and the non-spatial vMF regression model based on these prediction errors. The spatial vMF regression with $P=5$ demonstrates comparatively better performance among $P\in\{2,\ldots,7\}$.
 Here, true data is generated by setting $P=5$. It suggests that the optimal lag can be chosen based on the predictive performance for our real data analysis.

\begin{figure}
    \centering
    \includegraphics[width=0.45\textwidth]{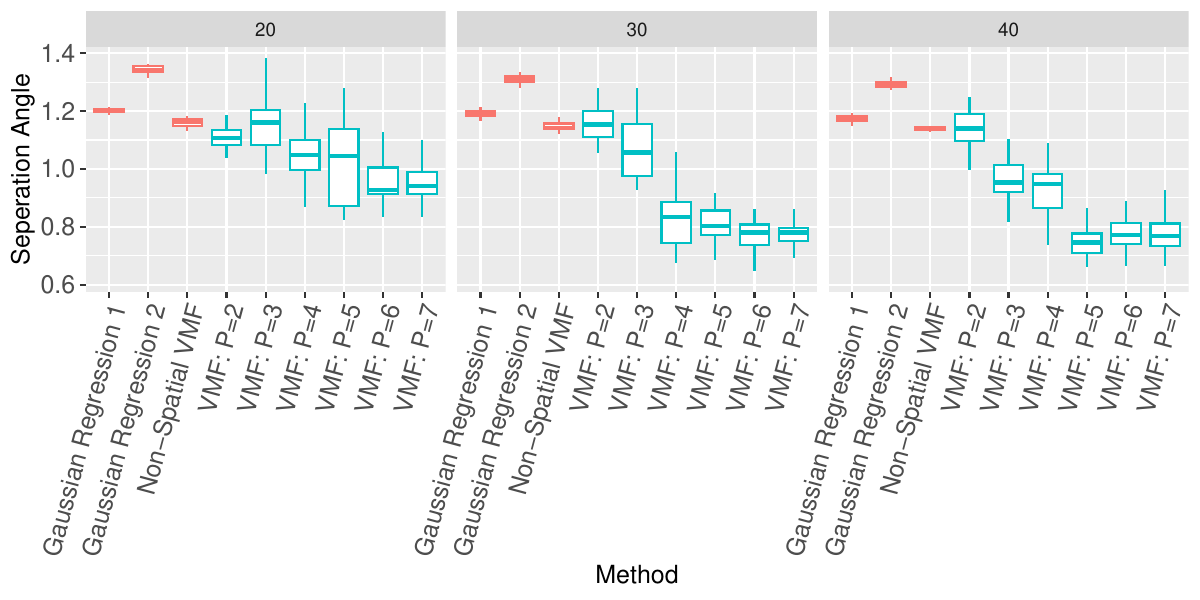}
    \includegraphics[width=0.45\textwidth]{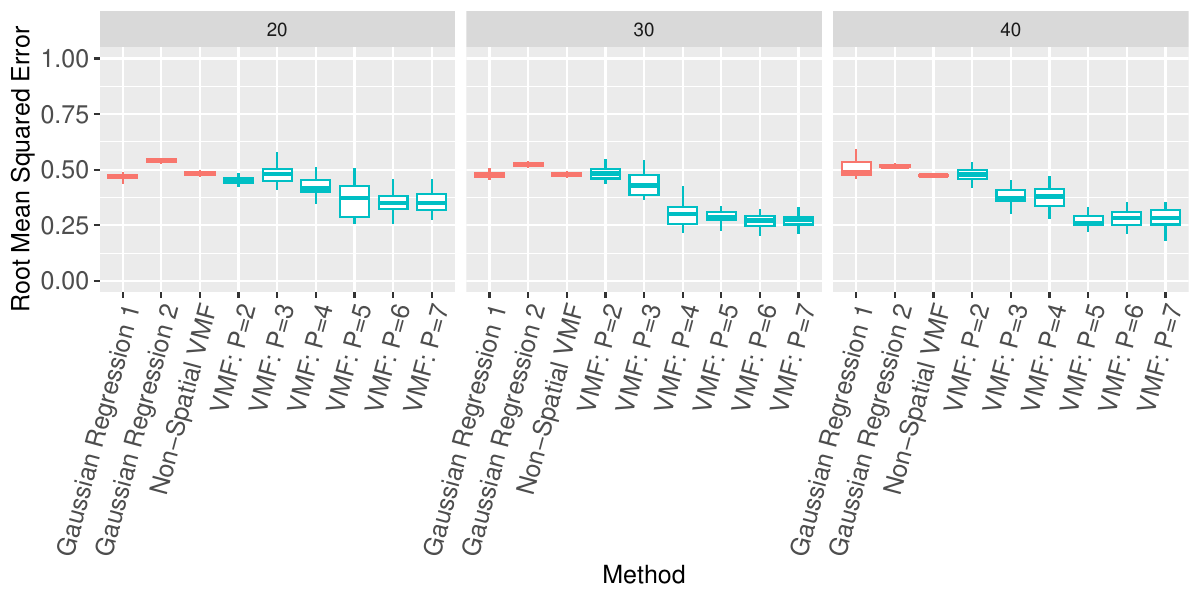}
    \caption{Comparison of prediction errors quantified in terms of separation angles (Left Panel) and root-mean-squared errors (Right Panel) between the observed and predicted means. For the proposed vMF model, errors for different choices of spatial lags from 2 to 7 are provided.}
    \label{fig:result}
\end{figure}

\subsection{Bayesian Angular Inference for ADNI Data}
\label{sec:analysis}
In this section, we utilize our Bayesian angular inference method to analyze data from the Alzheimer's Disease Neuroimaging Initiative (ADNI). We illustrate the effects of covariates via our proposed tangent-normal decomposition. Our techniques shed fresh light on the neurological progression of Alzheimer's disease. To maintain a specific emphasis in our study, we concentrate on examining the influence of the APOE-4 allele on brain anatomical structure. This analysis takes into account other covariates such as age, gender, and scores from the Mini-Mental State Examination (MMSE). We focus on the fornix in this section.

To achieve this goal, the \textit{typical} covariate vector is given as follows
\begin{equation}
\label{eq:male}
    \begin{aligned}
        \bar{\bm{X}}_{T}^{(g)}=[1, \bar{X}_{age}^{(g)},\bar{X}_{T,gender}=0, \bar{X}_{MMSE}^{(g)}, X_{T,APOE}=0]
    \end{aligned}
\end{equation}
or
\begin{equation}
\label{eq:female}
    \begin{aligned}
        \bar{\bm{X}}_{T}^{(g)}=[1, \bar{X}_{age}^{(g)},\bar{X}_{T,gender}=1, \bar{X}_{MMSE}^{(g)}, X_{T,APOE}=0]
    \end{aligned}
\end{equation}
where Equations (\ref{eq:male}) and (\ref{eq:female}) are the \textit{typical} covariate vectors of male and female, respectively. We consider the wild type of APOE-4 (i.e., $X_{T,APOE}=0$) as \textit{typical}. The terms $\bar{X}_{age}^{(g)}$ and $\bar{X}_{MMSE}^{(g)}$ are the averages of age and MMSE within group $g$. We decompose the gender specification into male and female, in order to provide the gender-specific inference. The corresponding \textit{specific} covariate vector is given as
\begin{equation}
\label{eq:male1}
    \begin{aligned}
        \bar{\bm{X}}_{S}^{(g)}=[1, \bar{X}_{age}^{(g)},\bar{X}_{T,gender}=0, \bar{X}_{MMSE}^{(g)}, X_{S,APOE}=1]
    \end{aligned}
\end{equation}
or
\begin{equation}
\label{eq:female1}
    \begin{aligned}
        \bar{\bm{X}}_{S}^{(g)}=[1, \bar{X}_{age}^{(g)},\bar{X}_{s,gender}=1, \bar{X}_{MMSE}^{(g)}, X_{S,APOE}=1],
    \end{aligned}
\end{equation}
where the APOE-4 status is considered as mutated (i.e., $X_{S,APOE}=1$).

To incorporate group-level analysis into our study, we use the superscript $(g)$ to denote group-specific terms. In this section, we focused on the investigation of the effect of APOE-4. The posterior $[m_{v}^{(g)} |\text{data}]$ implies the strength of the covariate effect and we consider $m_{v}^{(g)}>0.65$ as large impact. Therefore, we calculate the posterior probabilities for $Pr(m_{v}^{(g)}-0.65|\text{data})$.
As an exploratory analysis, we consider the voxels whose posterior z-value is larger than the 9th quantile of the posterior z-values over voxels and group as \textit{impacted} voxels.
The visualization of the ADNI group is in Figure \ref{fig:Directional_Effects}. The ADNI group's fornix anatomical structure is more impacted by the factor of APOE. Also, males have a higher laterality of the APOE effect than women. Male's left component of the fornix is more impacted by the factor of APOE.

\begin{figure}
    \centering
    \includegraphics[width=0.5\textwidth]{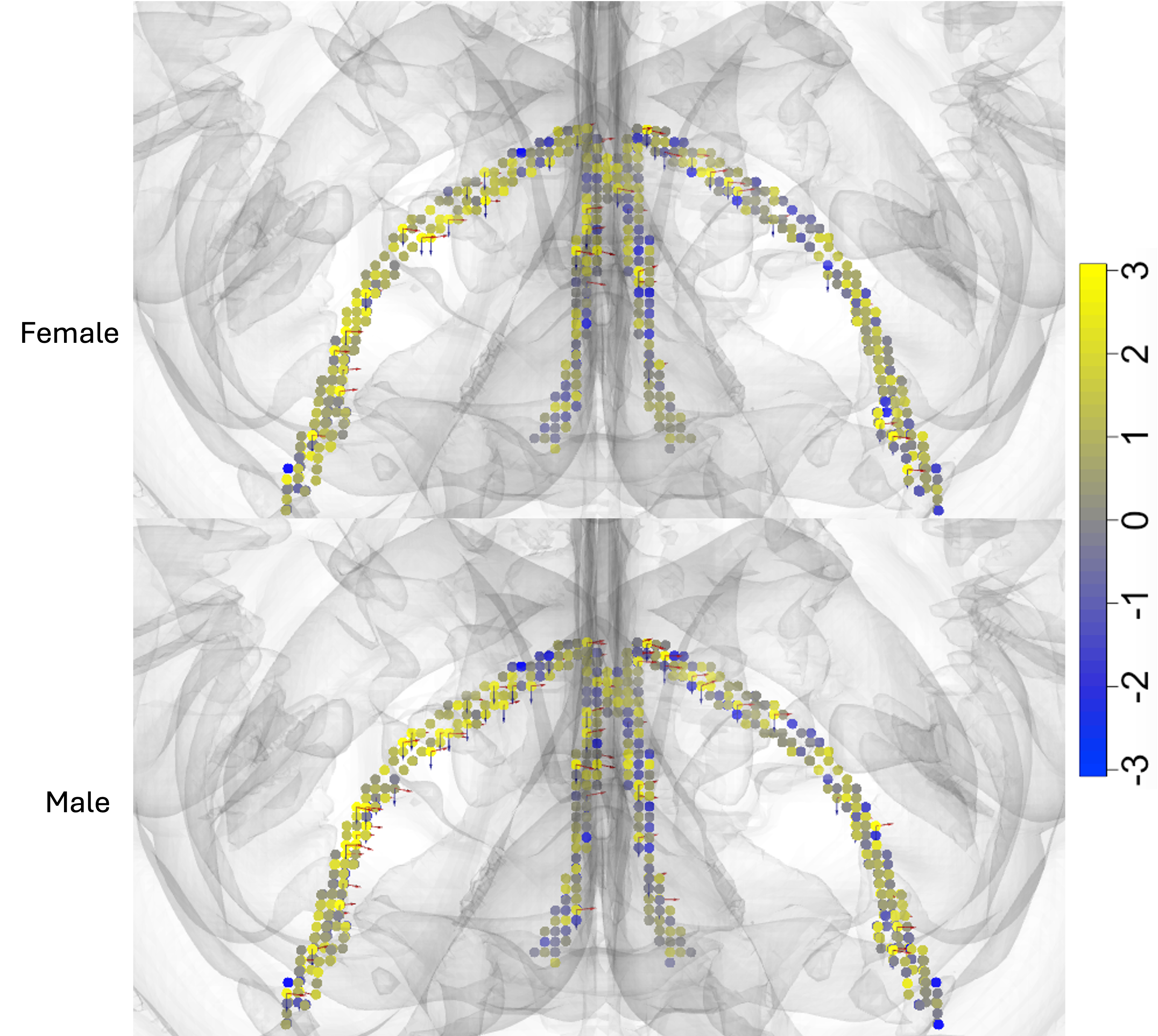}
    \caption{We present a tangent-normal decomposition of APOE-4 with a specific focus on the fornix as our region of interest. The upper panel is for females and the lower panel is for males. The colors are for the posterior z-scores. For \textit{impacted} voxels, the red arrows are for \textit{typical} mode directions and the blue arrows are for \textit{specific} mode directions.}
    \label{fig:Directional_Effects}
\end{figure}

\section{Conclusion and Discussion}
\label{sec:conclusion}
This paper presents a novel spatial generalized linear regression framework for modeling spatially varying directional data as an outcome. The framework leverages a vMF-distributed error model, shown to capture local, covariate-dependent variations accurately.
Given the directional nature of the data, we employ an autoregressive framework to capture spatial variations. Numerical evaluations, performed on both real and synthetic data, demonstrate our proposed method's superior performance. By applying our proposed Bayesian angular inference scheme to the ADNI data, we have gleaned several novel scientific insights.

Our current analysis is cross-sectional. However, the exploration of longitudinal changes, specific to our ADNI study, akin to \cite{wang2019hierarchical,kundu2019novel}, could be intriguing for identifying underlying mechanistic changes over time. Nonetheless, such covariate-dependent longitudinal research poses significant challenges. For instance, brain images and other biomarkers are often measured asynchronously \citep{li2022regression}. In future work, extending our generalized vMF regression to accommodate asynchronously collected longitudinal data would be an interesting direction to pursue. Furthermore, our proposed model can be seamlessly integrated into other popular inference frameworks, such as mediation analysis, providing additional avenues of interest.
In the context of analyzing DTI data, it will be interesting to incorporate fiber-shape information for more holistic analysis \citep{yeh2020shape}.
Future research will explore such a direction.

One of the most important aspects of our proposed model is the incorporation of spatial dependence. Thus, one may consider other modeling choices for the directional distribution and link functions proposed in the literature \citep{scealy2017directional, scealy2019scaled,paine2020spherical} and integrate spatial dependence. The proposed autoregressive spatial dependence can still be applied seamlessly to these other modeling choices. It is thus an important future research avenue to compare these different modeling choices holistically for spatially dependent directional data. Another direction of future research is to consider a scale mixture of vMF distribution to characterize the data, which will allow different subjects to have different $\kappa$'s. Another generalization is possible in the specification of $\kappa$. We have now fixed it for all the spatial locations. It might be interesting to allow $\kappa$ to vary slowly. Furthermore, unimodality might be a restrictive assumption.

\section*{Acknowledgment}
The authors gratefully acknowledge the Editor, Associate Editor, and reviewers for their valuable feedback. This work utilized data from the Alzheimer’s Disease Neuroimaging Initiative (ADNI) database (adni.loni.usc.edu). The ADNI investigators were responsible for the design and implementation of the study and/or provided data, but they did not participate in the analysis or writing of this manuscript. A full list of ADNI investigators is available at: \url{https://adni.loni.usc.edu/wp-content/uploads/how_to_apply/ADNI_Acknowledgement_List.pdf}. Zhou Lan thanks the salary support from the Sundry Fund of the Brigham and Women's Hospital Department of Radiology.

Data collection and sharing for this project was funded by the Alzheimer's Disease Neuroimaging Initiative (ADNI) (National Institutes of Health Grant U01 AG024904) and DOD ADNI (Department of Defense award number W81XWH-12-2-0012). ADNI is funded by the National Institute on Aging, the National Institute of Biomedical Imaging and Bioengineering, and through generous contributions from the following: AbbVie, Alzheimer's Association; Alzheimer's Drug Discovery Foundation; Araclon Biotech; BioClinica, Inc.; Biogen; Bristol-Myers Squibb Company; CereSpir, Inc.; Cogstate; Eisai Inc.; Elan Pharmaceuticals, Inc.; Eli Lilly and Company; EuroImmun; F. Hoffmann-La Roche Ltd and its affiliated company Genentech, Inc.; Fujirebio; GE
Healthcare; IXICO Ltd.; Janssen Alzheimer's Immunotherapy Research \& Development, LLC.; Johnson \& Johnson Pharmaceutical Research \& Development LLC.; Lumosity; Lundbeck; Merck \& Co., Inc.; Meso
Scale Diagnostics, LLC.; NeuroRx Research; Neurotrack Technologies; Novartis Pharmaceuticals Corporation; Pfizer Inc.; Piramal Imaging; Servier; Takeda Pharmaceutical Company; and Transition
Therapeutics. The Canadian Institutes of Health Research is providing funds to support ADNI clinical sites in Canada. Private sector contributions are facilitated by the Foundation for the National Institutes of Health ({\tt www.fnih.org}). The grantee organization is the Northern California Institute for Research and Education,
and the study is coordinated by the Alzheimer's Therapeutic Research Institute at the University of Southern California. ADNI data are disseminated by the Laboratory for Neuro Imaging at the University of Southern California.


\end{document}